\documentclass[a4paper,twoside,11pt]{article}
\usepackage[margin=2.9cm]{geometry}
\usepackage[affil-it]{authblk}
\usepackage{amsmath,amsfonts,amsthm,amssymb}
\usepackage[font={small}]{caption}
\usepackage[font={small}]{subfig}
\usepackage{color,graphicx}
\usepackage{times}
\usepackage{xspace}
\usepackage{paralist}
\usepackage{hyperref}

\newcommand{\diff}[2]{($#1,#2$)}

\newtheorem{lemma}{Lemma}
\newtheorem{theorem}{Theorem}

\title{The Maximum k-Differential Coloring Problem}

\author[1]{Michael A. Bekos%
\thanks{Electronic address: \texttt{bekos@informatik.uni-tuebingen.de}}}

\author[1]{Michael~Kaufmann%
\thanks{Electronic address: \texttt{mk@informatik.uni-tuebingen.de}}}

\author[2]{Stephen~G.~Kobourov%
\thanks{Electronic address: \texttt{kobourov@cs.arizona.edu}}}

\author[2]{Sankar~Veeramoni%
\thanks{Electronic address: \texttt{sankar@cs.arizona.edu}}}

\affil[1]{Wilhelm-Schickard-Institut f\"{u}r Informatik - Universit\"{a}t T\"{u}bingen, Germany}
\affil[2]{Department of Computer Science - University of Arizona, Tucson AZ, USA}

\date{}

\begin{document}
\maketitle

%=================================================================
\begin{abstract}
Given an $n$-vertex graph $G$ and two positive integers $d,k \in
\mathbb{N}$, the \diff{d}{kn}-differential coloring problem asks for
a coloring of the vertices of $G$ (if one exists) with distinct
numbers from $1$ to $kn$ (treated as \emph{colors}), such that the
minimum difference between the two colors of any adjacent vertices
is at least $d$. While it was known that the problem of determining
whether a general graph is \diff{2}{n}-differential colorable is
NP-complete, our main contribution is a complete characterization of
bipartite, planar and outerplanar graphs that admit
\diff{2}{n}-differential colorings. For practical reasons, we also
consider color ranges larger than $n$, i.e., $k > 1$. We show that
it is NP-complete to determine whether a graph admits a
\diff{3}{2n}-differential coloring. The same negative result holds
for the \diff{\lfloor2n/3\rfloor}{2n}-differential coloring problem,
even in the case where the input graph is planar.
\end{abstract}
%=================================================================

%=================================================================
\section{Introduction}
\label{intro}
%=================================================================
Several methods for visualizing relational datasets use a map
metaphor where objects, relations between objects and clusters are
represented as cities, roads and countries, respectively. Clusters
are usually represented by colored regions, whose boundaries are
explicitly defined. The $4$-coloring theorem states that four colors
always suffice to color any map such that neighboring countries have
distinct colors. However,  if not all countries of the map are
contiguous and the countries are not colored with unique colors, it
would be impossible to distinguish whether two regions with the same
color belong to the same country or to different countries. In order
to avoid such ambiguity, this necessitates the use of a unique color
for each country; see Fig.~\ref{fig:teaser}.

\begin{figure}[t!]
    \centering
    \begin{minipage}[b]{.49\textwidth}
        \centering
        \subfloat[\label{fig:teaser1}{Colored with random assignment of colors}]
        {\includegraphics[width=\textwidth]{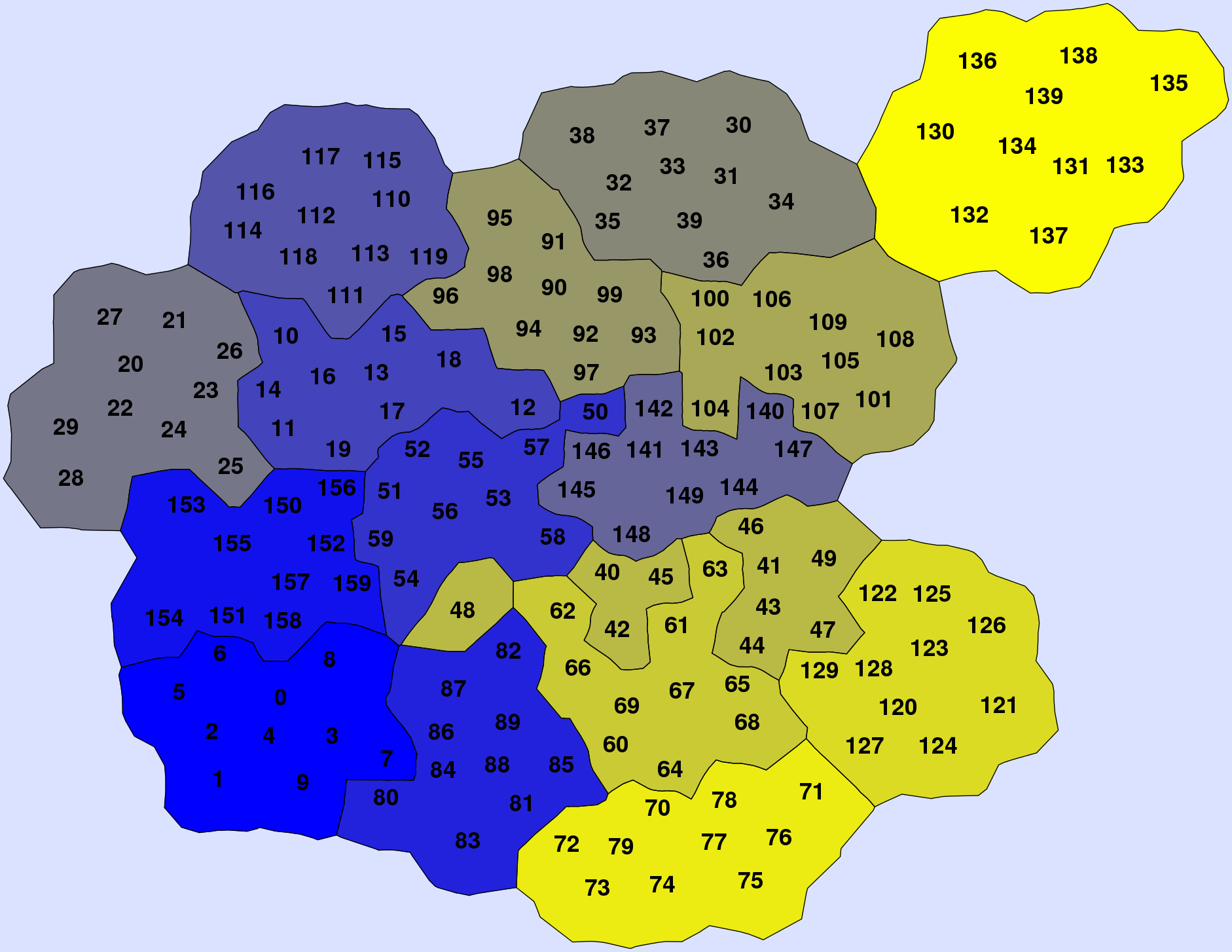}}
    \end{minipage}
    \begin{minipage}[b]{.49\textwidth}
        \centering
        \subfloat[\label{fig:teaser2}{Colored with max. differential coloring}]
        {\includegraphics[width=\textwidth]{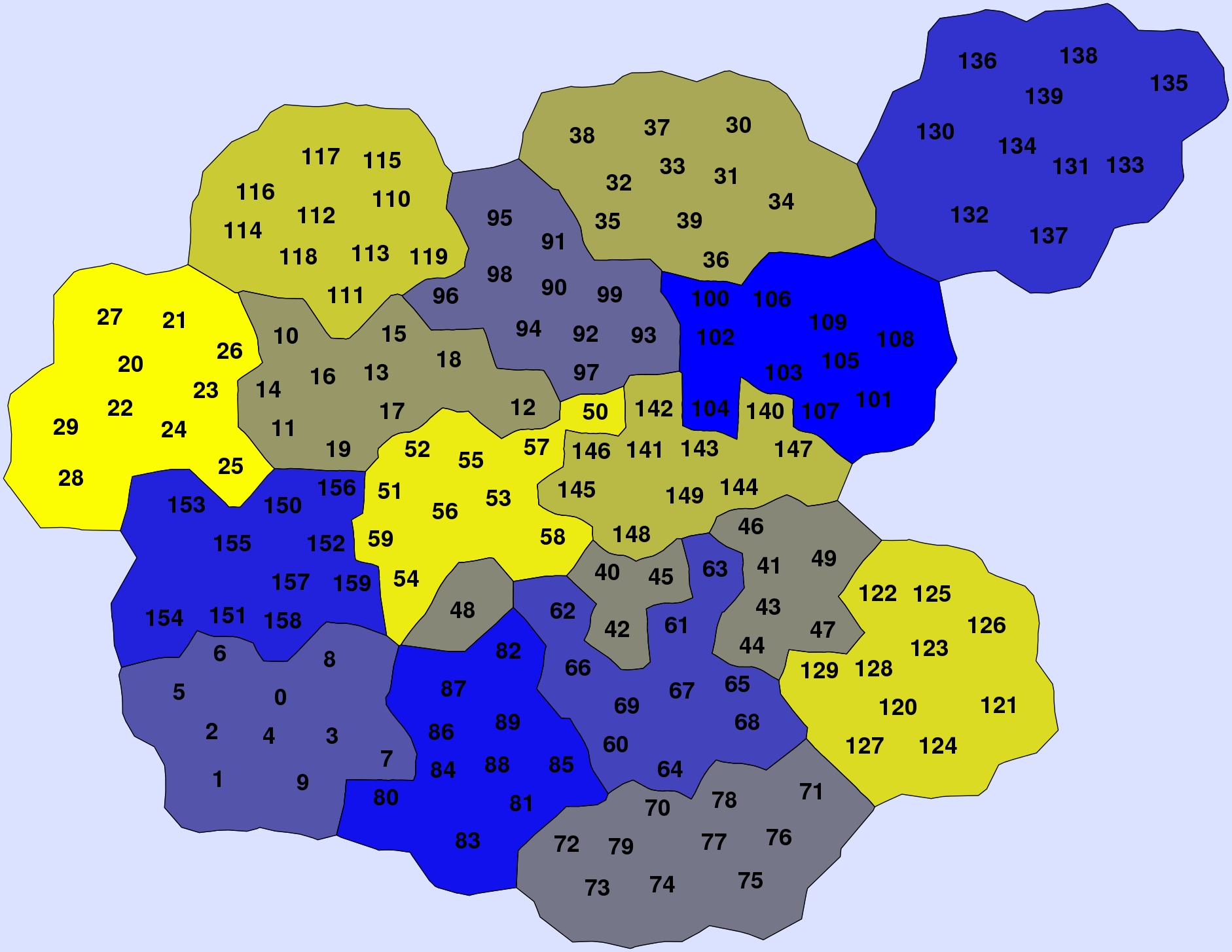}}
    \end{minipage}
    \caption{Illustration of a map colored using the same set of
    colors obtained by the linear interpolation of blue and yellow. There is one country
    in the middle containing the vertices 40-49 which is fragmented into three small regions.}
    \label{fig:teaser}
\end{figure}

However, it is not enough to just assign different colors to each
country. Although human perception of color is good and thousands of
different colors can be easily distinguished, reading a map can be
difficult due to color constancy and color context
effects~\cite{purves2003we}. Dillencourt et al.~\cite{GE07} define a
good coloring as one in which the colors assigned to the countries
are visually distinct while also ensuring that the colors assigned
to adjacent countries are as dissimilar as possible. However, not
all colors make suitable choices for coloring countries and a
``good'' color palette is often a gradation of certain map-like
colors~\cite{ColorBrewer2}. In more restricted scenarios, e.g., when
a map is printed in gray scale, or when the countries in a given
continent must use different shades of a predetermined color, the
color space becomes 1-dimensional.

This 1-dimensional fragmented map coloring problem is nicely
captured by the \emph{maximum differential coloring
problem}~\cite{abktree,leung1984,miller89,yixun2003dual}, which we
slightly generalize in this paper: Given a map, define the
\emph{country graph} $G=(V,E)$ whose vertices represent countries,
and two countries are connected by an edge if they share a
non-trivial geographic boundary. Given two positive integers $d,k
\in \mathbb{N}$, we say that $G$ is \diff{d}{kn}-differential
colorable if and only if there is a coloring of the $n$ vertices of
$G$ with distinct numbers from $1$ to $kn$ (treated as
\emph{colors}), so that the \emph{minimum color distance} between
adjacent vertices of $G$ is at least $d$. The \emph{maximum
k-differential coloring} problem asks for the largest value of $d$,
called the \emph{k-differential chromatic number} of $G$, so that
$G$ is \diff{d}{kn}-differential colorable. Note that the
traditional \emph{maximum differential coloring problem} corresponds
to $k=1$.

A natural reason to study the maximum k-differential coloring
problem for $k>1$ is that using more colors can help produce maps
with larger differential chromatic number. Note, for example, that a
star graph on $n$ vertices has 1-differential chromatic number (or
simply \emph{differential chromatic number}) one, whereas its
2-differential chromatic number is $n+1$. That is, by doubling the
number of colors used, we can improve the quality of the resulting
coloring by a factor of $n$. This is our main motivation for
studying the maximum k-differential coloring problem for $k>1$.

%=================================================================
\subsection{Related Work}
%=================================================================
The maximum differential coloring problem is a well-studied problem,
which dates back in 1984, when Leung et al.~\cite{leung1984}
introduced it under the name ``separation number'' and showed its
NP-completeness. It is worth mentioning though that the maximum
differential coloring problem is also known as ``dual
bandwidth''~\cite{yixun2003dual} and
``anti-bandwidth''~\cite{abktree}, since it is the complement of the
{\em bandwidth minimization problem}~\cite{Papadimitriou_1975}. Due
to the hardness of the problem, heuristics are often used for
coloring general graphs, e.g., LP-formulations~\cite{grasp}, memetic
algorithms~\cite{memetic} and spectral based methods~\cite{hu2010}.
The differential chromatic number is known only for special graph
classes, such as Hamming graphs~\cite{abhamming},
meshes~\cite{abrsstv}, hypercubes~\cite{abrsstv,Wang20091947},
complete binary trees~\cite{weili}, complete $m$-ary trees for odd
values of $m$~\cite{abktree}, other special types of
trees~\cite{weili}, and complements of interval graphs, threshold
graphs and arborescent comparability graphs~\cite{Isaak98powersof}.
Upper bounds on the differential chromatic number are given by Leung
et al.~\cite{leung1984} for connected graphs and by Miller and
Pritikin~\cite{miller89} for bipartite graphs. For a more detailed
bibliographic overview refer to~\cite{arxivplanar}.

%=================================================================
\subsection{Our Contribution}
%=================================================================
In Section~\ref{sec:preliminaries}, we present preliminary
properties and bounds on the $k$-differential chromatic number. One
of them guarantees that any graph is \diff{1}{n}-differential
colorable; an arbitrary assignment of distinct colors to the
vertices of the input graph guarantees a minimum color distance of
one (see Lemma~\ref{lem:prilim1}). So, the next reasonable question
to ask is whether a given graph is \diff{2}{n}-differential
colorable. Unfortunately, this is already an NP-complete problem
(for general graphs), since a graph is \diff{2}{n}-differential
colorable if and only if its complement has a Hamiltonian
path~\cite{leung1984}. This motivates the study of the
\diff{2}{n}-differential coloring problem for special classes of
graphs. In Section~\ref{sec:difflabel21}, we present a complete
characterization of bipartite, outer-planar and planar graphs that
admit \diff{2}{n}-differential colorings.

In Section~\ref{sec:gennpcom}, we double the number of available
colors. As any graph is \diff{2}{2n}-differential colorable (due to
Lemma~\ref{lem:prilim1}; Section~\ref{sec:preliminaries}), we study
the \diff{3}{2n}-differential coloring problem and we prove that it
is NP-complete for general graphs (Theorem~\ref{thm:reduction};
Section~\ref{sec:gennpcom}). We also show that testing whether a
given graph is \diff{k+1}{kn}-differential colorable is NP-complete
(Theorem~\ref{thm:reduction3}; Section~\ref{sec:gennpcom}). On the
other hand, all planar graphs are \diff{\lfloor
n/3\rfloor+1}{2n}-differential colorable (see
Lemma~\ref{lem:mcolorable}; Section~\ref{sec:preliminaries}) and
testing whether a given planar graph is
\diff{\lfloor2n/3\rfloor}{2n}-differential colorable is shown to be
NP-complete (Theorem~\ref{thm:np-proof};
Section~\ref{sec:gennpcom}). In Section~\ref{sec:ILP}, we provide a
simple ILP-formulation for the maximum k-differential coloring
problem and experimentally compare the optimal results obtained by
the ILP formulation for $k=1$ and $k=2$ with  GMap, which is a
heuristic based on spectral methods developed by Hu et
al.~\cite{Gansner_Hu_Kobourov_2009_gmap}. We conclude in
Section~\ref{sec:conclusion} with open problems and future work.

%=================================================================
\section{Preliminaries}
\label{sec:preliminaries}
%=================================================================
The maximum k-differential coloring problem can be easily reduced to
the ordinary differential coloring problem as follows: If $G$ is an
$n$-vertex graph that is input to the maximum k-differential
coloring problem, create a disconnected graph $G'$ that contains all
vertices and edges of $G$ plus $(k-1) \cdot n$ isolated vertices.
Clearly, the k-differential chromatic number of $G$ is equal to the
1-differential chromatic number of $G'$. A drawback of this
approach, however, is that few results are known for the ordinary
differential coloring problem, when the input is a disconnected
graph. In the following, we present some immediate upper and lower
bounds on the $k$-differential chromatic number for connected
graphs.

%=================================================================
\begin{lemma}
The $k$-differential chromatic number of a connected graph is at
least $k$.
\label{lem:prilim1}
\end{lemma}
%=================================================================
\begin{proof}
Let $G$ be a connected graph on $n$ vertices. It suffices to prove
that $G$ is \diff{k}{kn}-differential colorable. Indeed, an
arbitrary assignment of distinct colors from the set $\{k, 2k,
\ldots, kn\}$ to the vertices of $G$ guarantees a minimum color
distance of $k$.
\end{proof}

%=================================================================
\begin{lemma}
The $k$-differential chromatic number of a connected graph $G =
(V,E)$ on $n$ vertices is at most $ \lfloor \frac{n}{2} \rfloor +
(k-1)n$.
\label{lem:prilim2}
\end{lemma}
%=================================================================
\begin{proof}
The proof is a straightforward generalization of the proof of Yixun
and Jinjiang~\cite{yixun2003dual} for the ordinary maximum
differential coloring problem. One of the vertices of $G$ has to be
assigned with a color in the interval $[\lceil \frac{n}{2}
\rceil,\lceil \frac{n}{2} \rceil + (k-1)n]$, as the size of this
interval is $(k-1)n + 1$ and there can be only $(k-1)n$ unassigned
colors. Since $G$ is connected, that vertex must have at least one
neighbor which (regardless of its color) would make the difference
along that edge at most $kn - \lceil \frac{n}{2} \rceil = \lfloor
\frac{n}{2} \rfloor + (k-1)n$.
\end{proof}

%=================================================================
\begin{lemma}
The $k$-differential chromatic number of a connected $m$-colorable
graph $G = (V,E)$ on $n$ vertices is at least
$\lfloor\frac{(k-1)n}{m-1}\rfloor + 1$.
\label{lem:mcolorable}
\end{lemma}
%=================================================================
\begin{proof}
Let $C_i \subseteq V$ be the set of vertices of $G$ with color $i$
and $c_i$ be the number of vertices with color $i$, $i = 1, \ldots,
m$. We can show that $G$ is \diff{\lfloor\frac{(k-1)n}{m-1}\rfloor +
1}{kn}-differential colorable by coloring the vertices of $C_i$ with
colors from the following set: $ [~(\sum_{j = 1}^{i-1}c_j) + 1 +
(i-1)\lfloor\frac{(k-1)n}{m-1}\rfloor,~(\sum_{j = 1}^{i}c_j) +
(i-1)\lfloor\frac{(k-1)n}{m-1}\rfloor~] $
\end{proof}

%=================================================================
\section{The (2,n)-Differential Coloring Problem}
\label{sec:difflabel21}
%=================================================================

In this section, we provide a complete characterization of
\begin{inparaenum}[(i)]
\item bipartite graphs,
\item outerplanar graphs and
\item planar graphs
\end{inparaenum}
that admit \diff{2}{n}-differential coloring. Central to our
approach is a result of Leung et al.~\cite{leung1984} who showed
that a graph $G$ has \diff{2}{n}-differential coloring if and only
if the complement $G^c$ of $G$ is Hamiltonian. As a consequence, if
the complement of $G$ is disconnected, then $G$ has no
\diff{2}{n}-differential coloring.

In order to simplify our notation scheme, we introduce the notion of
\emph{ordered differential coloring} (or simply \emph{ordered
coloring}) of a graph, which is defined as follows. Given a graph
$G=(V,E)$ and a sequence $S_1 \rightarrow S_2 \rightarrow \ldots
\rightarrow S_k$ of $k$ disjoint subsets of $V$, such that $\cup_{i
= 1}^{k}S_i = V$, an \emph{ordered coloring} of $G$ implied by the
sequence $S_1 \rightarrow S_2 \rightarrow \ldots \rightarrow S_k$ is
one in which the vertices of $S_i$ are assigned colors from
$(\sum_{j = 1}^{i-1}|S_j|) + 1$ to $\sum_{j = 1}^{i}|S_j|$, $i=1, 2,
\ldots, k$.

%=================================================================
\begin{theorem}
A bipartite graph admits a \diff{2}{n}-differential coloring if and
only if it is not a complete bipartite graph.
\label{lem:bipartite}
\end{theorem}
%=================================================================
\begin{proof}
Let $G = (V,E)$ be an $n$-vertex bipartite graph, with $V = V_1\cup
V_2$, $V_1 \cap V_2 = \emptyset$ and $E \subseteq V_1 \times V_2$.
If $G$ is a complete bipartite graph, then its complement is
disconnected. Therefore, $G$ does not admit a
\diff{2}{n}-differential coloring. Now, assume that $G$ is not
complete bipartite. Then, there exist at least two vertices, say $u
\in V_1$ and $v \in V_2$, that are not adjacent, i.e., $(u,v) \notin
E$. Consider the ordered coloring of $G$ implied by the sequence
$V_1 \setminus \{u\}\rightarrow \{u\} \rightarrow \{v\} \rightarrow
V_2 \setminus \{v\}$. As $u$ and $v$ are not adjacent, it follows
that the color difference between any two vertices of $G$ is at
least two. Hence, $G$ admits a \diff{2}{n}-differential coloring.
\end{proof}

%=================================================================
\begin{lemma}
An outerplanar graph with $n \geq 6$ vertices, that does not contain
$K_{1,n-1}$ as a subgraph, admits a $3$-coloring, in which each
color set contains at least $2$ vertices.
\label{lem:outerplanarcoloring}
\end{lemma}
%=================================================================
\begin{proof}
Let $G=(V,E)$ be an outerplanar graph with $n \geq 6$ vertices, that
does not contain $K_{1,n-1}$ as a subgraph. As $G$ is outerplanar,
it admits a $3$-coloring~\cite{PS86}. Let $C_i \subseteq V$ be the
set of vertices of $G$ with color $i$ and $c_i$ be the number of
vertices with color $i$, that is $c_i=|C_i|$, for $i=1, 2, 3$.
Without loss of generality let $c_1 \le c_2 \le c_3$. We further
assume that each color set contains at least one vertex, that is
$c_i \geq 1$, $i=1,2,3$. If there is no set with less than $2$
vertices, then the lemma clearly holds. Otherwise, we distinguish
three cases:

\begin{enumerate}
\item[Case 1:]
$c_1 = c_2 = 1$ and $c_3 \ge 4$. Without loss of generality assume
that $C_1 = \{a\}$ and $C_2 = \{b\}$. As $G$ is outerplanar,
vertices $a$ and $b$ can have at most $2$ common neighbors. On the
other hand, since $G$ has at least $6$ vertices, there exists at
least one vertex, say $c \in C_3$, which is not a common neighbor of
$a$ and $b$. Without loss of generality assume that $(b,c) \notin
E$. Then, vertex $c$ can be colored with color $2$. Therefore, we
derive a new $3$-coloring of $G$ for which we have that $c_1 = 1$,
$c_2 = 2$ and $c_3 \ge 3$.

\item[Case 2:]
$c_1 = 1$, $c_2 = 2$ and $c_3 \ge 3$: Without loss of generality
assume that $C_1 = \{a\}$ and $C_2 = \{b,b'\}$. First, consider the
case where there exists at least one vertex, say $c \in C_3$, which
is not a neighbor of vertex $a$. In this case, vertex $c$ can be
colored with color $1$ and a new $3$-coloring of $G$ is derived with
$c_1 = c_2 = 2$ and $c_3 \ge 3$, as desired. Now consider the more
interesting case, where vertex $a$ is a neighbor of all vertices of
$C_3$. As $G$ does not contain $K_{1,n-1}$ as a subgraph, either
vertex $b$ or vertex $b'$ is not a neighbor of vertex $a$. Without
loss of generality let that vertex be $b$, that is $(a,b) \notin E$.
As $G$ is outerplanar, vertices $a$ and $b'$ can have at most $2$
common neighbors. Since $G$ has at least $6$ vertices and vertex $a$
is a neighbor of all vertices of $C_3$, there exist at least one
vertex, say $c \in C_3$, which is not adjacent to vertex $b'$, that
is $(b',c) \notin E$. Therefore, we can color vertex $c$ with color
$2$ and vertex $b$ with color $1$ and derive a new $3$-coloring of
$G$ for which we have that $c_1 = c_2 = 2$ and $c_3 \ge 2$, as
desired.

\item[Case 3:]
$c_1 = 1$, $c_2 \ge 3$ and $c_3 \ge 3$: Without loss of generality
assume that $C_1 = \{a\}$. Then, there exists at least one vertex,
say $c \in C_2 \cup C_3$, which is not a neighbor of vertex $a$. In
this case, vertex $c$ can be colored with color $1$ and a new
$3$-coloring of $G$ is derived with $c_1 = c_2 = 2$ and $c_3 \ge 3$,
as desired.
\end{enumerate}
\end{proof}

%=================================================================
\begin{lemma}
Let  $G=(V,E)$ be an outerplanar graph and let $V'$ and $V''$ be two
disjoint subsets of $V$, such that $|V'| \ge 2$ and $|V''| \ge 3$.
Then, there exist two vertices $u \in V'$ and $v \in V''$, such
that $(u,v) \notin E$.
\label{lem:outerplanar2vertices}
\end{lemma}
%=================================================================
\begin{proof}
The proof follows from the fact that an outerplanar graph is
$K_{2,3}$ free.
\end{proof}

%=================================================================
\begin{theorem}
An outerplanar graph with $n \ge 8$ vertices has
\diff{2}{n}-differential coloring if and only if it does not contain
$K_{1,n-1}$ as subgraph. \label{thm:outerplanar2diff}
\end{theorem}
%=================================================================
\begin{proof}
Let $G=(V,E)$ be an outerplanar graph with $n \ge 8$ vertices. If
$G$ contains $K_{1,n-1}$ as subgraph, then the complement $G^c$ of
$G$ is disconnected. Therefore, $G$ does not admit a
\diff{2}{n}-differential coloring. Now, assume that $G$ does not
contain $K_{1,n-1}$ as subgraph. By
Lemma~\ref{lem:outerplanarcoloring}, it follows that $G$ admits a
$3$-coloring, in which each color set contains at least two
vertices. Let $C_i \subseteq V$ be the set of vertices with color
$i$ and $c_i=|C_i|$, for $i=1,2,3$, such that $2 \leq c_1 \le c_2
\le c_3$. We distinguish the following cases:

\begin{enumerate}
\item[Case 1:]
$c_1 = 2$, $c_2 = 2$, $c_3 \geq 4$. Since $|C_1| = 2$ and $|C_3| \ge
4$, by Lemma~\ref{lem:outerplanar2vertices} it follows that there
exist two vertices $a \in C_1$ and $c \in C_3$, such that $(a,c)
\notin E$. Similarly, since $|C_2| = 2$ and $|C_3 \setminus \{c\}|
\ge 3$, by Lemma~\ref{lem:outerplanar2vertices} it follows that
there exist two vertices $b \in C_2$ and $c' \in C_3$, such that $c
\neq c'$ and $(b,c') \notin E$; see
Fig.~\ref{fig:out2diff1}-\ref{fig:out2diff2}.

\item[Case 2:]
$c_1 \geq 2$, $c_2 \geq 3$, $c_3 \geq 3$. Since $|C_1| = 2$ and
$|C_3| \ge 3$, by Lemma~\ref{lem:outerplanar2vertices} it follows
that there exist two vertices $a \in C_1$ and $c \in C_3$, such that
$(a,c) \notin E$. Similarly, since $|C_2| \geq 3$ and $|C_3
\setminus \{c\}| \ge 2$, by Lemma~\ref{lem:outerplanar2vertices} it
follows that there exist two vertices $b \in C_2$ and $c' \in C_3$,
such that $c \neq c'$ and $(b,c') \notin E$; see
Fig.~\ref{fig:out2diff3}-\ref{fig:out2diff4}.
\end{enumerate}

\begin{figure}[t!]
  \centering
  \begin{minipage}[b]{.23\textwidth}
    \centering
    \subfloat[\label{fig:out2diff1}]
    {\includegraphics[width=\textwidth,page=3]{images/outerplanar2diff}}
  \end{minipage}
  \begin{minipage}[b]{.23\textwidth}
    \centering
    \subfloat[\label{fig:out2diff2}]
    {\includegraphics[width=\textwidth,page=4]{images/outerplanar2diff}}
  \end{minipage}
  \begin{minipage}[b]{.23\textwidth}
    \centering
    \subfloat[\label{fig:out2diff3}]
    {\includegraphics[width=\textwidth,page=1]{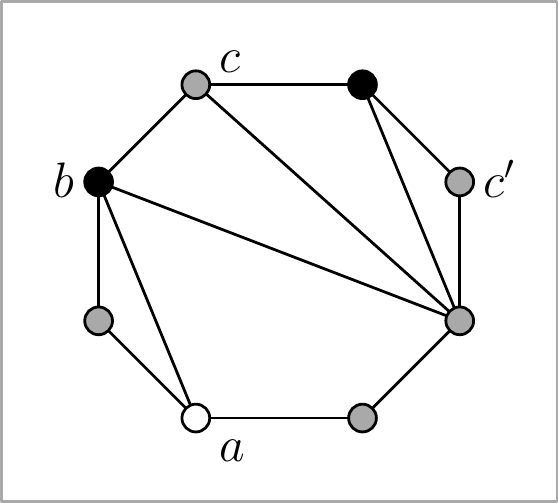}}
  \end{minipage}
    \begin{minipage}[b]{.23\textwidth}
    \centering
    \subfloat[\label{fig:out2diff4}]
    {\includegraphics[width=\textwidth,page=2]{images/outerplanar2diff}}
  \end{minipage}
  \caption{
  (a) An outerplanar graph colored with $3$ colors, white, black and grey (Case~$1$ of Thm.~\ref{thm:outerplanar2diff}), and,
  (b)~its \diff{2}{n}-differential coloring.
  (c) Another outerplanar graph also colored with $3$ colors, white, black and grey (Case~$2$ of Thm.~\ref{thm:outerplanar2diff}), and,
  (d)~its \diff{2}{n}-differential coloring.}
  \label{fig:out2diff}
\end{figure}

For both cases, consider the ordered coloring implied by the
sequence $C_1 \setminus \{a\} \rightarrow \{a\} \rightarrow \{c\}
\rightarrow C_3 \setminus \{c,c'\} \rightarrow \{c'\} \rightarrow
\{b\} \rightarrow C_2 \setminus \{b\}$. As $(a,c) \notin E$ and
$(b,c') \notin E$, it follows that the color difference between any
two vertices of $G$ is at least two. Hence, $G$ admits a
\diff{2}{n}-differential coloring.
\end{proof}

%=================================================================
\begin{lemma}
A planar graph with $n \geq 36$ vertices, that does not contain as
subgraphs $K_{1,1,n-3}$, $K_{1,n-1}$ and $K_{2,n-2}$, admits a
$4$-coloring, in which two color sets contain at least $2$ vertices
and the remaining two at least $5$ vertices.
\label{lem:planarcoloring}
\end{lemma}
%=================================================================
\begin{proof}
Let $G=(V,E)$ be a planar graph with $n \geq 36$ vertices, that does
not contain as subgraphs $K_{1,1,n-3}$, $K_{1,n-1}$ and $K_{2,n-2}$.
As $G$ is planar, it admits a $4$-coloring~\cite{diestel2000graph}.
Let $C_i \subseteq V$ be the set of vertices of $G$ with color $i$
and $c_i$ be the number of vertices with color $i$, that is
$c_i=|C_i|$, for $i=1,2,3,4$. Without loss of generality let $c_1
\le c_2 \le c_3 \le c_4$. We further assume that each color set
contains at least one vertex, that is $c_i \geq 1$, $i=1,2,3,4$. We
distinguish the following cases:

\begin{enumerate}
\item[Case 1:]
$c_1 = 1$, $c_2 = 1$, $c_3 \le 7$. Since $G$ has at least $36$
vertices, it follows that $c_4 \ge 27$. Observe that a vertex of
$C_4$ cannot be assigned a color other than $4$ if and only if it is
adjacent to at least one vertex in $C_1$, one vertex in $C_2$ and
one vertex in $C_3$. Let $C^*_4 \subseteq C_4$ be the set of such
vertices and $c^*_4=|C^*_4|$. We claim that $c^*_4 \leq 14$. To
prove the claim, we construct an auxiliary bipartite graph
$G_{\text{aux}}=(V_{\text{aux}},E_{\text{aux}})$ where
$V_{\text{aux}} = C_1 \cup C_2 \cup C_3 \cup C^*_4$ and
$E_{\text{aux}} = E \cap (C^*_4 \times (C_1 \cup C_2 \cup C_3))$.
Clearly, $G_{\text{aux}}$ is planar, since it is subgraph of $G$.
Since all vertices of $C^*_4$ have degree at least $3$, it holds
that $3c^*_4 \leq |E_{\text{aux}}|$. On the other hand, it is known
that a planar bipartite graph on $n$ vertices cannot have more that
$2n-4$ edges~\cite{lovatz}. Therefore, $|E_{\text{aux}}| \leq 2(c_1
+ c_2 + c_3 + c^*_4)-4$, which implies that $3c^*_4 \leq 2(c_1 + c_2
+ c_3 + c^*_4)-4$. As $c_1 + c_2 + c_3 \leq 9$, it follows that our
claim indeed holds. A a consequence, we can change the color of
several vertices being currently in $C_4$ and obtain a new
$4$-coloring of $G$ in which $C_4$ has exactly $14$ vertices. This
implies that the number of vertices in $C_1$, $C_2$ and $C_3$ is at
least $22$ or equivalently that one out of $C_1$, $C_2$ and $C_3$
must contain strictly more that $7$ vertices, say $C_3$. So, in the
new coloring it holds that $c_1 \geq 1$, $c_2 \ge 1$ and $c_3 \ge
8$; a case that is covered in the following.

\item[Case 2:]
$c_1 = 1$, $c_2 = 1$, $c_3 \ge 8$. Assume w.l.o.g. that $C_1 =
\{a\}$ and $C_2 = \{b\}$. Since $G$ does not contain $K_{2,n-2}$ as
subgraph, there exists at least one vertex, say $c \in C_3 \cup
C_4$, which is not neighboring either $a$ or $b$, say  $b$:  $(b,c)
\notin E$. Then, vertex $c$ can be safely colored with color $2$. In
the new coloring, it holds that $c_1 = 1$, $c_2 = 2$ and $c_3 \geq
7$; refer to Case~4.

\item[Case 3:]
$c_1 = 1$, $c_2 = 2$, $c_3 \le 6$. Since $G$ has at least $36$
vertices, it follows that $c_4 \ge 27$. Now, observe that $c_1 + c_2
+ c_3 \leq 9$. Hence, following similar arguments as in Case~1, we
can prove that there is a $4$-coloring of $G$ in which $C_4$ has
exactly $14$ vertices. So, again the number of vertices in $C_1$,
$C_2$ and $C_3$ is at least $22$ and consequently at least one out
of $C_1$, $C_2$ and $C_3$ must contain strictly more that $7$
vertices, say $C_3$. In the new coloring, it holds that $c_1 \geq
1$, $c_2 \geq 2$ and $c_3 \ge 8$; refer to Case~6.

\item[Case 4:]
$c_1 = 1$, $c_2 = 2$, $c_3 \ge 7$. Assume w.l.o.g. that $C_1 =
\{a\}$ and $C_2 = \{b,b'\}$. First, consider the case where there
exists at least one vertex, say $c \in  C_3 \cup C_4$, which is not
neighboring with vertex $a$. In this particular case, vertex $c$ can
be colored with color $1$ and a new $4$-coloring of $G$ is derived
for which it holds that $c_1 = 2$, $c_2 = 2$ and $c_3 \geq 6$, as
desired. Now consider the more interesting case, where vertex $a$ is
neighboring with all vertices of $C_3 \cup C_4$. As $G$ does not
contain $K_{1,n-1}$ as subgraph, vertex $a$ is not a neighbor of
vertex $b$ or vertex $b'$ or both. In the latter case, vertex $a$
can be colored with color $2$. Therefore, if we select two arbitrary
vertices of $C_3$ and color them with color $1$, we derive a new
$4$-coloring of $G$ for which it holds that $c_1 = 2$, $c_2 = 2$ and
$c_3 \ge 5$, as desired. On the other hand, if vertex $a$ is
neighboring exactly one out of $b$ and $b'$, say w.l.o.g. $b$, that
is $(a,b) \in E$ and $(a,b') \notin E$, then there exists a vertex
$c \in C_3 \cup C_4$ such that $(b,c) \notin E$; as otherwise
$K_{a,b,C_3 \cup C_4}$ forms a $K_{1,1,n-3}$. So, if we color vertex
$b'$ with color $1$ and vertex $c$ with color $2$, we obtain a new
coloring of $G$  for which it holds that $c_1 = 2$, $c_2 = 2$ and
$c_3 \ge 6$, as desired.

\item[Case 5:]
$c_1 = 1$, $c_2 \geq 3$, $c_3 \geq 3$. Assume w.l.o.g. that $C_1 =
\{a\}$. As $G$ does not contain $K_{1,n-1}$ as subgraph, there is a
vertex $a' \in C_2 \cup C_3 \cup C_4$ which is not neighboring with
$a$. Hence, $a'$ can be colored with color $1$. In the new coloring,
it holds that $c_1 \ge 2$, $c_2 \ge 2$, $c_3 \ge 2$; refer to
Case~$6$.

\item[Case 6:]
$c_1 \geq 2$, $c_2 \geq 2$, $c_3 \leq 4$. Since $G$ has at least
$36$ vertices, it follows that $c_4 \ge 24$. Since $c_1 + c_2 + c_3
\le 12$, similarly to Case~1 we can prove that there is a $4$-coloring
of $G$ in which $C_4$ has exactly $20$ vertices. So, the number of
vertices in $C_1$, $C_2$ and $C_3$ is at least $16$ and consequently
at least one out of $C_1$, $C_2$ and $C_3$ must contain strictly
more that $5$ vertices, say $C_3$. In the new coloring, it holds
that $c_1 \geq 2$, $c_2 \geq 2$ and $c_3 \ge 6$, as desired.
\end{enumerate}
From the above case analysis, it follows that $G$ has a
$4$-coloring, in which two color sets contain at least $2$ vertices
and the remaining two at least $5$ vertices.
\end{proof}

%=================================================================
\begin{lemma}
Let  $G=(V,E)$ be a planar graph and let $V'$ and $V''$ be two
disjoint subsets of $V$, such that $|V'| \ge 3$ and $|V''| \ge 3$.
Then, there exists two vertices $u \in V'$ and $v \in V''$, such
that $(u,v) \notin E$.
\label{lem:planar2vertices}
\end{lemma}
%=================================================================
\begin{proof}
The proof follows directly from the fact that a planar graph does
not contain $K_{3,3}$ as subgraph.
\end{proof}

%=================================================================
\begin{theorem}
A planar graph with $n \ge 36$ vertices has a
\diff{2}{n}-differential coloring if and only if it does not contain
as subgraphs $K_{1,1,n-3}$, $K_{1,n-1}$ and $K_{2,n-2}$.
\label{thm:planar2diff}
\end{theorem}
%=================================================================
\begin{proof}
Let $G=(V,E)$ be an $n$-vertex planar graph, with $n \ge 36$. If $G$
contains $K_{l,n-l}$ as subgraph, then the complement $G^c$ of $G$
is disconnected, $l=1,2\ldots, \lceil n/2 \rceil$. Therefore, $G$
has no \diff{2}{n}-differential coloring. Note, however, that since
$G$ is planar, it cannot contain $K_{l,n-l}$, $l=3,4\ldots, \lceil
n/2 \rceil$, as subgraph. On the other hand, if $G$ contains
$K_{1,1,n-3}$ as subgraph, then $G^c$ has two vertices of degree one
with a common neighbor. Hence, $G^c$ is not Hamiltonian and as a
consequence $G$ has no \diff{2}{n}-differential coloring, as well.

\begin{figure}[t!]
  \centering
  \begin{minipage}[b]{.24\textwidth}
        \centering
        \subfloat[\label{fig:2diff1}]
    {\includegraphics[width=\textwidth,page=1]{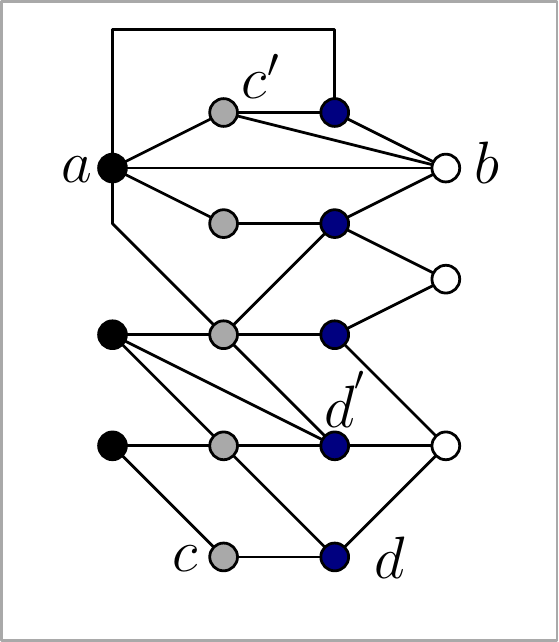}}
    \end{minipage}
  \begin{minipage}[b]{.24\textwidth}
        \centering
        \subfloat[\label{fig:2diff2}]
    {\includegraphics[width=\textwidth,page=2]{images/planar2difftight}}
    \end{minipage}
  \begin{minipage}[b]{.24\textwidth}
        \centering
        \subfloat[\label{fig:2diff3}]
    {\includegraphics[width=\textwidth,page=3]{images/planar2difftight}}
    \end{minipage}
    \begin{minipage}[b]{.24\textwidth}
        \centering
        \subfloat[\label{fig:2diff4}]
    {\includegraphics[width=\textwidth,page=4]{images/planar2difftight}}
    \end{minipage}
    \caption{
    (a)~A planar graph colored with $4$ colors, black ($C_1$), white ($C_2$), grey ($C_3$) and blue ($C_4$); Case~$1$ of Thm.~\ref{thm:planar2diff}, and,
    (b)~its \diff{2}{n}-differential coloring.
    (c)~Another planar graph also colored with $4$ colors, black, grey, blue and white; Case~$2$ of Thm.~\ref{thm:planar2diff}, and,
    (d)~its \diff{2}{n}-differential coloring.}
    \label{fig:2diff}
\end{figure}

Now, assume that $G$ does not contain as subgraphs $K_{1,1,n-3}$,
$K_{1,n-1}$ and $K_{2,n-2}$. By Lemma~\ref{lem:planarcoloring}, it
follows that $G$ admits a $4$-coloring, in which two color sets
contain at least $2$ vertices and the remaining two at least $5$
vertices. Let $C_i \subseteq V$ be the set of vertices with color
$i$ and $c_i = |C_i|$, for $i = 1,2,3,4$, such that $c_1 \leq c_2
\leq c_3 \leq c_4$, $c_1,c_2 \geq 2$ and $c_3,c_4 \geq 5$. We
distinguish the following cases:

\begin{enumerate}

\item[Case 1:]
$c_1 \geq 3$, $c_2 \geq 3$, $c_3 \geq 5$, $c_4 \geq 5$. Since $|C_1|
\geq 3$ and $|C_3| \ge 5$, by Lemma~\ref{lem:planar2vertices} it
follows that there exist two vertices $a \in C_1$ and $c \in C_3$,
such that $(a,c) \notin E$. Similarly, since $|C_2| \geq 3$ and
$|C_4| \ge 5$, by Lemma~\ref{lem:planar2vertices} it follows that
there exist two vertices $b \in C_2$ and $d \in C_4$, such that
$(b,d) \notin E$. Finally, since $|C_3 \setminus \{c\}| \geq 4$ and
$|C_4 \setminus \{d\}| \ge 4$, by Lemma~\ref{lem:planar2vertices} it
follows that there exist two vertices $c' \in C_3$ and $d' \in C_4$,
such that $c \neq c'$, $d \neq d'$ and $(c',d') \notin E$. Now
consider the ordered coloring implied by the sequence $C_1 \setminus
\{a\} \rightarrow \{a\} \rightarrow \{c\} \rightarrow C_3 \setminus
\{c,c'\} \rightarrow \{c'\} \rightarrow \{d'\} \rightarrow C_4
\setminus \{d,d'\} \rightarrow \{d\} \rightarrow \{b\} \rightarrow
C_2 \setminus \{b\}$. As $(a,c) \notin E$, $(b,d) \notin E$ and
$(c',d') \notin E$, it follows that the color difference between any
two vertices of G is at least two. Hence, $G$ has a
\diff{2}{n}-differential coloring.

\item[Case 2:]
$c_1 =2$, $c_2 \ge 4$, $c_3 \ge 5$, $c_4 \ge 5$. Since $G$ does not
contain $K_{2,n-2}$ as subgraph, there exists two vertices, say $a
\in C_1$ and $b \in C_2 \cup C_3 \cup C_4$ such that $(a,b) \notin
E$. Assume w.l.o.g. that $b \in C_2$. Since $|C_2 \setminus \{b\}|
\geq 3$ and $|C_3| \ge 5$, by Lemma~\ref{lem:planar2vertices} it
follows that there exist two vertices $b' \in C_2$ and $c \in C_3$,
such that $b' \neq b$ and $(b',c) \notin E$. Similarly, since $|C_3
\setminus \{c\}| \geq 4$ and $|C_4| \ge 5$, by
Lemma~\ref{lem:planar2vertices} it follows that there exist two
vertices $c' \in C_3$ and $d \in C_4$, such that $c' \neq c$ and
$(c',d) \notin E$. Similarly, to the previous case the ordered
coloring implied by the sequence $C_1 \setminus \{a\} \rightarrow
\{b\} \rightarrow C_2 \setminus \{b\} \rightarrow \{b'\} \rightarrow
\{c\} \rightarrow C_3 \setminus \{c,c'\} \rightarrow \{c'\}
\rightarrow \{d\} \rightarrow C_4 \setminus \{d\}$ guarantees that
$G$ has a \diff{2}{n}-differential coloring.

\item[Case 3:] $c_1 =2$, $c_2 = 3$, $c_3 \geq 5$, $c_4 \geq 5$.
Assume w.l.o.g. that $C_1=\{a,a'\}$ and $C_2=\{b,b',b''\}$. We
distinguish two sub cases:

\begin{itemize}[-]
\item The subgraph of $G$ induced by $C_1 \cup C_2$ is $K_{2,3}$,
that is $C_1 \times C_2 \subseteq E$. Since $G$ does not contain
$K_{2,n-2}$ as subgraph, there exists at least one vertex of $C_1$,
say vertex $a$, that is not neighboring with a vertex of $C_3 \cup
C_4$, say vertex $c \in C_3$, that is $(a,c) \notin E$. Since $|C_2|
= 3$ and $|C_4| \ge 5$, by Lemma~\ref{lem:planar2vertices} it
follows that there exist a vertex of $C_2$, say vertex $b$, and a
vertex of $C_4$, say vertex $d$, that are not adjacent, that is
$(b,d) \notin E$. Similarly, since $|C_3 \setminus \{c\}| \geq 4$
and $|C_4 \setminus \{d\}| \ge 4$, by
Lemma~\ref{lem:planar2vertices} it follows that there exist two
vertices $c' \in C_3$ and $d' \in C_4$, such that $c \neq c'$, $d
\neq d'$ and $(c',d') \notin E$. So, in the ordered coloring implied
by the sequence $C_1 \setminus \{a\} \rightarrow \{a\} \rightarrow
\{c\} \rightarrow C_3 \setminus \{c,c'\} \rightarrow \{c'\}
\rightarrow \{d'\} \rightarrow C_4 \setminus \{d,d'\} \rightarrow
\{d\} \rightarrow \{b\} \rightarrow C_2 \setminus \{b\}$, it holds
that the color difference between any two vertices of $G$ is greater
or equal to two. Hence, $G$ has a \diff{2}{n}-differential coloring.
\item The subgraph of $G$ induced by $C_1 \cup C_2$ is not
$K_{2,3}$. Assume w.l.o.g. that $(a,b) \notin E$. Since $|C_3| \geq
5$ and $|C_4| \ge 5$, by Lemma~\ref{lem:planar2vertices} it follows
that there exist a vertex of $C_3$, say vertex $c$, and a vertex of
$C_4$, say vertex $d$, that are not adjacent, that is $(c,d) \notin
E$. Similarly, since $|C_3 \setminus \{c\}| \geq 4$ and
$|\{a',b',b''\}| = 3$, by Lemma~\ref{lem:planar2vertices} it follows
that there exist a vertex of $C_3 \setminus \{c\}$, say vertex $c'$,
and a vertex of $\{a',b',b''\}$, say vertex $x$, that are not
adjacent, that is $(x,c') \notin E$. First consider the case where
$x=a'$. In this case, the ordered coloring implied by the sequence
$\{b'\} \rightarrow \{b''\} \rightarrow \{b\} \rightarrow \{a\}
\rightarrow \{a'\} \rightarrow \{c'\} \rightarrow C_3 \setminus
\{c,c'\} \rightarrow \{c\} \rightarrow \{d\} \rightarrow C_4
\setminus \{d\}$ guarantees that $G$ has a \diff{2}{n}-differential
coloring. Now consider the case, where $x \in \{b',b'\}$. As both
cases are symmetric, we assume w.l.o.g. that $x=b'$. In this case,
the ordered coloring implied by the sequence $\{a'\} \rightarrow
\{a\} \rightarrow \{b\} \rightarrow \{b''\} \rightarrow \{b'\}
\rightarrow \{c'\} \rightarrow C_3 \setminus \{c,c'\} \rightarrow
\{c\} \rightarrow \{d\} \rightarrow C_4 \setminus \{d\}$ guarantees
that $G$ has a \diff{2}{n}-differential coloring.
\end{itemize}

\item[Case 4:]
$c_1 =2$, $c_2 = 2$, $c_3 \ge 5$, $c_4 \ge 5$. Assume w.l.o.g. that
$C_1=\{a,a'\}$ and $C_2=\{b,b'\}$. Again, we distinguish two sub
cases:

\begin{itemize}
\item The subgraph of $G$ induced by $C_1 \cup C_2$ is $K_{2,2}$,
that is $C_1 \times C_2 \subseteq E$. Since $G$ does not contain
$K_{2,n-2}$ as subgraph, there exists at least one vertex of $C_1$,
say vertex $a$, that is not neighboring with a vertex of $C_3 \cup
C_4$, say vertex $c \in C_3$, that is $(a,c) \notin E$. Similarly,
there exists a vertex of $C_2$, say vertex $b$, that is not
neighboring with a vertex of $C_3 \cup C_4$, say vertex $w$, that is
$(b,w) \notin E$.

First, consider the case where the subgraph of $G$ induced by $C_2
\cup C_4$ is not $K_{2,c_4}$ that is $w \in C_4$ . Since $|C_3
\setminus \{c\}| \geq 4$ and $|C_4 \setminus \{w\}| \ge 4$, by
Lemma~\ref{lem:planar2vertices} it follows that there exist a vertex
of $C_3$, say vertex $c'$, and a vertex of $C_4$, say vertex $d$,
that are not adjacent, that is $c \neq c'$, $w \neq d$ and $(c',d)
\notin E$. In this case, the ordered coloring implied by the
sequence $C_1\setminus \{a\} \rightarrow \{a\} \rightarrow \{c\}
\rightarrow C_3 \setminus \{c,c'\} \rightarrow \{c'\} \rightarrow
\{d\} \rightarrow C_4\setminus\{d,w\} \rightarrow \{w\} \rightarrow
\{b\} \rightarrow C_2 \setminus \{b\}$ guarantees that $G$ has a
\diff{2}{n}-differential coloring.

Now consider the more interesting case where the subgraph of $G$
induced by $C_2 \cup C_4$ is $K_{2,c_4}$, that is $w \in C_3$. We
distinguish two sub cases:
\begin{itemize}
\item[a.] $w \neq c$.
Since $|C_3 \setminus \{c,w\}| \geq 3$ and $|C_4| \ge 5$, by
Lemma~\ref{lem:planar2vertices} it follows that there exist a vertex
of $C_3$, say vertex $p$, and a vertex of $C_4$, say vertex $q$,
that are not adjacent, that is $c \neq w \neq p$, and $(p,q) \notin
E$. Similarly, since $|C_3 \setminus \{p,w\}| \geq 3$ and $|C_4
\setminus \{q\}| \ge 4$, by Lemma~\ref{lem:planar2vertices} it
follows that there exist a vertex of $C_3$, say vertex $p'$, and a
vertex of $C_4$, say vertex $q'$, that are not adjacent, that is $p
\neq w \neq p'$, $q \neq q'$ and $(p',q') \notin E$. If $p' \neq c$,
the ordered coloring implied by the sequence $C_1\setminus \{a\}
\rightarrow \{a\} \rightarrow \{c\} \rightarrow \{p'\} \rightarrow
\{q'\} \rightarrow C_4 \setminus \{q,q'\} \rightarrow \{q\}
\rightarrow \{p\} \rightarrow C_3\setminus\{p,p',c,w\} \rightarrow
\{w\} \rightarrow \{b\} \rightarrow C_2 \setminus \{b\}$ guarantees
that $G$ has a \diff{2}{n}-differential coloring. If $p' = c$, the
ordered coloring implied by the sequence $C_1\setminus \{a\}
\rightarrow \{a\} \rightarrow \{p'\} \rightarrow \{q'\} \rightarrow
C_4 \setminus \{q,q'\} \rightarrow \{q\} \rightarrow \{p\}
\rightarrow C_3\setminus\{p,p',w\} \rightarrow \{w\} \rightarrow
\{b\} \rightarrow C_2 \setminus \{b\}$ guarantees that $G$ has a
\diff{2}{n}-differential coloring.
\item[b.] $w = c$.
Since the subgraph of $G$ induced by $C_2 \cup C_4$ is $K_{2,c_4}$
and $|C_2 \cup \{a\}| = 3$ and $|C_4| \ge 4$, by
Lemma~\ref{lem:planar2vertices} it follows that there exist a vertex
of $C_4$, say vertex $q$, not adjacent to vertex $a$, that is $(a,q)
\notin E$. Similarly, since $|C_3 \setminus \{c\}| \ge  4$ and $|C_4
\setminus \{q\}| \ge 4$, by Lemma~\ref{lem:planar2vertices} it
follows that there exist a vertex of $C_4$, say vertex $q'$, and a
vertex of $C_3$, say vertex $p'$, that are not adjacent, that is $q'
\neq q$, $p' \neq c$ and $(p',q') \notin E$. Then the ordered
coloring implied by the sequence $C_1\setminus \{a\} \rightarrow
\{a\} \rightarrow \{q\} \rightarrow C_4 \setminus \{q,q'\}
\rightarrow \{q'\} \rightarrow \{p'\} \rightarrow
C_3\setminus\{p',c\} \rightarrow \{c\} \rightarrow \{b\} \rightarrow
C_2 \setminus \{b\}$ guarantees that $G$ has a
\diff{2}{n}-differential coloring.

\end{itemize}
\item The subgraph of $G$ induced by $C_1 \cup C_2$ is not
$K_{2,2}$. Assume w.l.o.g. that $(a,b) \notin E$.

First, consider the case where the subgraph of $G$ induced by
$\{a',b'\} \cup C_3 \cup C_4$ is $K_{2,c_3+c_4}$ that is $\{a',b'\}
\times C_3 \cup C_4 \subseteq E$. Since $|\{a,a',b'\}| =  3$ and
$|C_3| \ge 4$, by Lemma~\ref{lem:planar2vertices} it follows that
there exist a vertex of $C_3$, say vertex $c$, not adjacent to
vertex a, that is $(a,c) \notin E$. Similarly since, $|\{b,a',b'\}|
=  3$ and $|C_4| \ge 4$, by Lemma~\ref{lem:planar2vertices} it
follows that there exist a vertex of $C_4$, say vertex $d$, not
adjacent to vertex b, that is $(b,d) \notin E$. Since $|C_3
\setminus \{c\}| \geq 3$ and $|C_4 \setminus \{d\}| \ge 4$, by
Lemma~\ref{lem:planar2vertices} it follows that there exist a vertex
of $C_3$, say vertex $c'$, and a vertex of $C_4$, say vertex $d'$,
that are not adjacent, that is $c \neq c'$, $d \neq d'$ and $(c',d')
\notin E$. Then the ordered coloring implied by the sequence
$C_1\setminus \{a\} \rightarrow \{a\} \rightarrow \{c\} \rightarrow
C_3 \setminus \{c,c'\} \rightarrow \{c'\} \rightarrow \{d'\}
\rightarrow C_4\setminus\{d,d'\} \rightarrow \{d\} \rightarrow \{b\}
\rightarrow C_2 \setminus \{b\}$ guarantees that $G$ has a
\diff{2}{n}-differential coloring.

Now we consider the case where the subgraph of $G$ induced by
$\{a',b'\} \cup C_3 \cup C_4$ is not $K_{2,c_3+c_4}$. There exist a
vertex of $C_3 \cup C_4$, say vertex $c$, and a vertex of
$\{a',b'\}$, say vertex $p$, that are not adjacent, that is $(p,c)
\notin E$. Assume w.l.o.g $c \in C_3$. Since $|C_3 \setminus \{c\}|
\geq 3$ and $|C_4| \ge 4$, by Lemma~\ref{lem:planar2vertices} it
follows that there exist a vertex of $C_3$, say vertex $c'$, and a
vertex of $C_4$, say vertex $d$, that are not adjacent, that is $c
\neq c'$ and $(c',d') \notin E$. If $p = b'$ the ordered coloring
implied by the sequence $\{a'\} \rightarrow \{a\} \rightarrow \{b\}
\rightarrow \{b'\} \rightarrow \{c\} \rightarrow
C_3\setminus\{c,c'\} \rightarrow \{c'\} \rightarrow \{d'\}
\rightarrow C_4 \setminus \{d'\}$ guarantees that $G$ has a
\diff{2}{n}-differential coloring. If $p = a'$ the ordered coloring
implied by the sequence $\{b'\} \rightarrow \{b\} \rightarrow \{a\}
\rightarrow \{a'\} \rightarrow \{c\} \rightarrow
C_3\setminus\{c,c'\} \rightarrow \{c'\} \rightarrow \{d'\}
\rightarrow C_4 \setminus \{d'\}$ guarantees that $G$ has a
\diff{2}{n}-differential coloring.
\end{itemize}
\end{enumerate}
From the above case analysis, it follows that $G$ is
\diff{2}{n}-differential colorable, as desired.
\end{proof}

%=================================================================
\section{NP-completeness Results}
\label{sec:gennpcom}
%=================================================================

In this section, we prove that the \diff{3}{2n}-differential
coloring problem is NP-complete. Recall that all graphs are
\diff{2}{2n}-differential colorable due to Lemma~\ref{lem:prilim1}.

%=================================================================
\begin{theorem}
Given a graph $G=(V,E)$ on $n$ vertices, it is NP-complete to
determine whether $G$ has a \diff{3}{2n}-differential coloring.
\label{thm:reduction}
\end{theorem}
%=================================================================
\begin{proof}
The problem is clearly in NP, since a non-deterministic algorithm
needs only to guess an assignment of distinct colors (out of $2n$
available colors) to the vertices of the graph and then it is easy
to verify in polynomial time whether this assignment corresponds to
a differential coloring of color difference at least $3$.

In order to prove that the problem is NP-hard, we employ a reduction
from the \diff{3}{n}-differential coloring problem, which is known
to be NP-complete~\cite{leung1984}. More precisely, let $G=(V,E)$ be
an instance of the \diff{3}{n}-differential coloring problem, i.e.,
graph $G$ is an $n$-vertex graph with vertex set $V = \{v_1,v_2,
\ldots, v_n\}$. We will construct a new graph $G'$ with $n' = 2n$
vertices, so that $G'$ is \diff{3}{2n'}-differential colorable if
and only if $G$ is \diff{3}{n}-differential colorable; see
Fig.~\ref{fig:generalnp}.

Graph $G'=(V',E')$ is constructed by attaching $n$ new vertices to
$G$ that form a clique; see the gray colored vertices of
Fig.~\ref{fig:genrealnp2}. That is, $V'= V \cup U$, where $U =
\{u_1,u_2, \ldots, u_n\}$ and $(u,u') \in E'$ for any pair of
vertices $u$ and $u' \in U$. In addition, for each pair of vertices
$v \in V$ and $u \in U$ there is an edge connecting them in $G'$,
that is $(v,u) \in E'$. In other words,
\begin{inparaenum}[(i)]
\item the subgraph, say $G_U$, of $G'$ induced by $U$ is complete and
\item the bipartite graph, say $G_{U \times V}$, with bipartition $V$ and $U$ is also complete.
\end{inparaenum}
Observe that $G'$ is the join of $G$ and $K_n$.

\begin{figure}[t!]
    \centering
    \begin{minipage}[b]{.32\textwidth}
        \centering
        \subfloat[\label{fig:genrealnp1}{Instance $G=(V,E)$}]
        {\includegraphics[width=\textwidth,page=1]{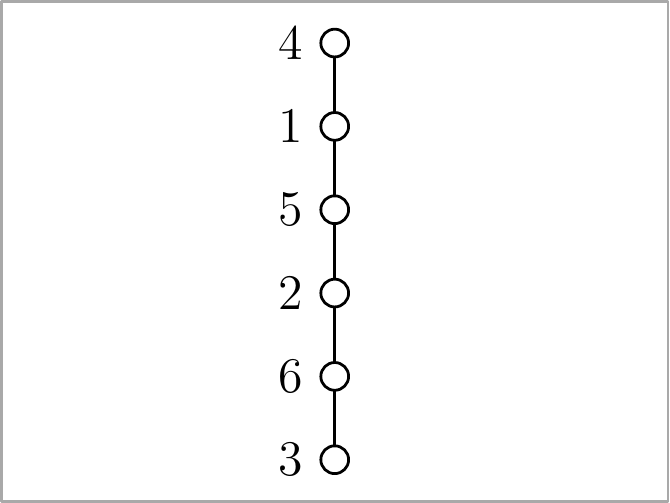}}
    \end{minipage}
    \begin{minipage}[b]{.32\textwidth}
        \centering
        \subfloat[\label{fig:genrealnp2}{Instance $G'=(V',E')$}]
        {\includegraphics[width=\textwidth,page=2]{images/generalnp}}
    \end{minipage}
    \hfill
    \caption{
    (a)~An instance of the \diff{3}{n}-differential coloring problem for $n=6$;
    (b)~An instance of the \diff{3}{2n'}-differential coloring problem constructed based on graph $G$.}
    \label{fig:generalnp}
\end{figure}

First, suppose that $G$ has a \diff{3}{n}-differential coloring and
let $l: V \to \{1, \ldots, n\}$ be the respective coloring. We
compute a coloring $l': V' \to \{1, \ldots, 4n\}$ of $G'$ as
follows:
\begin{inparaenum}[(i)]
\item $l'(v) = l(v)$, for all $v \in V' \cap V$ and
\item $l'(u_i)=n+3i$, $i=1,2,\ldots,n$.
\end{inparaenum}
Clearly, $l'$ is a \diff{3}{2n'}-differential coloring of $G'$.

Now, suppose that $G'$ is \diff{3}{2n'}-differential colorable and
let $l': V' \to \{1, \ldots, 2n'\}$ be the respective coloring
(recall that $n' = 2n$). We next show how to compute the
\diff{3}{n}-differential coloring for $G$. Without loss of
generality, let $V = \{v_1, \dots v_{n}\}$ contain the vertices of
$G$, such that $l'(v_1) < \ldots < l'(v_{n})$, and $U = \{u_1, \dots
u_{n}\}$ contains the newly added vertices of $G'$, such that
$l'(u_1) < \ldots < l'(u_{n})$. Since $G_U$ is complete, it follows
that the color difference between any two vertices of $U$ is at
least three. Similarly, since $G_{U \times V}$ is complete
bipartite, the color difference between any two vertices of $U$ and
$V$ is also at least three. We claim that $l'$ can be converted to
an equivalent \diff{3}{2n'}-differential coloring for $G'$, in which
all vertices of $V$ are colored with numbers from $1$ to $n$, and
all vertices of $U$ with numbers from $n+3$ to $4n$.

Let $U'$ be a maximal set of vertices $\{u_1, \dots, u_j\} \subseteq
U$ so that there is no vertex $v \in V$ with $l'(u_1)< l'(v) <
l'(u_j)$. If $U' = U$ and $l'(v) < l'(u_1), \forall v \in V$, then
our claim trivially holds. If $U' = U$ and $l'(v) > l'(u_j), \forall
v \in V$, then we can safely recolor all the vertices in $V'$ in the
reverse order, resulting in a coloring that complies with our claim.
Now consider the case where $U' \subsetneq U$. Then, there is a
vertex $v_k \in V$ s.t. $l'(v_k) - l'(u_j) \ge 3$. Similarly, we
define $V' = \{v_k, \ldots, v_l \in V\}$ to be a maximal set of
vertices of $V$, so that $l'(v_k) < \ldots <l'(v_l)$ and there is no
vertex $u \in U$ with $l'(v_k) < l'(u) <l'(v_l)$. Then, we can
safely recolor all vertices of $U' \cup V'$, such that:
\begin{inparaenum}[(i)]
\item the relative order of the colors of $U'$ and $V'$ remains unchanged,
\item the color distance between $v_l$ and $u_1$ is at least three, and
\item the colors of $U'$ are strictly greater than the ones of $V'$.
\end{inparaenum}
Note that the color difference between $u_j$ and $u_{j+1}$ and
between $v_{k-1}$ and $v_k$ is at least three after recoloring,
i.e., $l'(u_{j+1}) - l'(u_j) \ge 3$ and $l'(v_{k}) - l'(v_{k-1}) \ge
3$. If we repeat this procedure until $U'=U$, then the resulting
coloring complies with our claim. Thus, we obtain a
\diff{3}{n}-differential coloring $l$ for $G$ by assigning $l(v) =
l'(v), \forall v \in V$.
\end{proof}

%=================================================================
\begin{theorem}
Given a graph $G=(V,E)$ on $n$ vertices, it is NP-complete to
determine whether $G$ has a \diff{k+1}{kn}-differential coloring.
\label{thm:reduction3}
\end{theorem}
%=================================================================
\begin{proof}
Based on an instance $G=(V,E)$ of the \diff{k+1}{n}-differential
coloring problem, which is known to be NP-complete~\cite{leung1984},
construct a new graph $G'=(V',E')$ with $n' = kn$ vertices, by
attaching $n(k-1)$ new vertices to $G$, as in the proof of
Theorem~\ref{thm:reduction}. Then, using a similar argument as
above, we can show that $G$ has a \diff{k+1}{n}-differential
coloring if and only if $G'$ has a \diff{k+1}{kn'}-differential
coloring.
\end{proof}

The NP-completeness of $2$-differential coloring in
Theorem~\ref{thm:reduction} was about general graphs. Next, we
consider the complexity of the problem for planar graphs. Note that
from Lemma~\ref{lem:prilim2} and Lemma~\ref{lem:mcolorable}, it
follows that the $2$-differential chromatic number of a planar graph
on $n$-vertices is between $\lfloor\frac{n}{3}\rfloor + 1$ and
$\lfloor\frac{3n}{2}\rfloor$ (a planar graph is $4$-colorable). The
next theorem shows that testing whether a planar graph is
\diff{\lfloor2n/3\rfloor}{2n}-differential colorable is NP-complete.
Since this problem can be reduced to the general $2$-differential
chromatic number problem, it is NP-complete to determine the
$2$-differential chromatic number even for planar graphs.

%=================================================================
\begin{theorem}
Given an $n$-vertex planar graph $G=(V,E)$, it is NP-complete
to determine if $G$ has a \diff{\lfloor2n/3\rfloor}{2n}-differential
coloring.
\label{thm:np-proof}
\end{theorem}
%=================================================================
\begin{proof}
The problem is clearly in NP; a non-deterministic algorithm needs
only to guess an assignment of distinct colors (out of $2n$
available colors) to the vertices of the graph and then it is easy
to verify in polynomial time whether this assignment corresponds to
a differential coloring of minimum color difference at least
$\lfloor2n/3\rfloor$.

To prove that the problem is NP-hard, we employ a reduction from the
well-known $3$-coloring problem, which is NP-complete for planar
graphs~\cite{Garey:1979:CIG:578533}.  Let $G=(V,E)$ be an instance
of the $3$-coloring problem, i.e., $G$ is an $n$-vertex planar
graph. We will construct a new planar graph $G'$ with $n'=3n$
vertices, so that $G'$ is
\diff{\lfloor2n'/3\rfloor}{2n'}-differential colorable if and only
if $G$ is $3$-colorable.

Graph $G'=(V',E')$ is constructed by attaching a path $v \rightarrow
v_1 \rightarrow v_2$ to each vertex $v \in V$ of $G$; see
Fig.~\ref{fig:instance}-\ref{fig:cunstruction}. Hence, we can
assume that $V'=V \cup V_1 \cup V_2$, where $V$ is the vertex set of
$G$, $V_1$ contains the first vertex of each 2-vertex path and $V_2$
the second vertices. Clearly, $G'$ is a planar graph on $n' = 3n$ vertices.
Since $G$ is a subgraph of $G'$, $G$ is $3$-colorable if $G'$ is
$3$-colorable. On the other hand, if $G$ is $3$-colorable, then $G'$
is also $3$-colorable: for each vertex $v \in V$, simply color its
neighbors $v_1$ and $v_2$ with two distinct colors different from
the color of $v$. Next, we show that $G'$ is $3$-colorable if and
only if $G'$ has a \diff{\lfloor2n'/3\rfloor}{2n'}-differential
coloring.

First assume that $G'$ has a \diff{\lfloor
2n'/3\rfloor}{2n'}-differential coloring and let $l: V' \to \{1,
\ldots , 2n'\}$ be the respective coloring. Let $u \in V'$ be a
vertex of $G'$. We assign a color $c(u)$ to $u$ as follows:
$c(u)=i$, if $2(i-1)n + 1 \leq l(u) \leq 2in$, $i = 1,2,3$. Since
$l$ is a \diff{\lfloor2n'/3\rfloor}{2n'}-differential coloring, no
two vertices with the same color are adjacent. Hence, coloring $c$
is a $3$-coloring for $G'$.

\begin{figure}[t!]
    \centering
    \begin{minipage}[b]{.22\textwidth}
        \centering
        \subfloat[\label{fig:instance}{Instance $G$}]
        {\includegraphics[width=\textwidth,page=1]{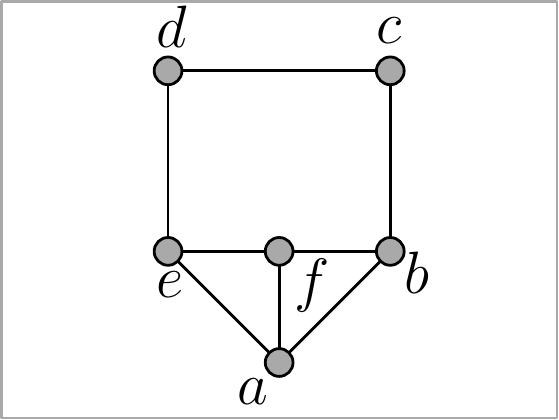}}
    \end{minipage}
    \begin{minipage}[b]{.22\textwidth}
        \centering
        \subfloat[\label{fig:cunstruction}{Graph $G'$}]
        {\includegraphics[width=\textwidth,page=2]{images/reduction}}
    \end{minipage}
    \begin{minipage}[b]{.54\textwidth}
        \centering
        \subfloat[\label{fig:coloring}{Differential coloring}]
        {\includegraphics[width=\textwidth,page=3]{images/reduction}}
    \end{minipage}
    \caption{
    (a)~An instance of the $3$-coloring problem;
    (b)~An instance of the \diff{\lfloor2n'/3\rfloor}{2n'}-differential coloring problem constructed based on graph $G$;
    (c)~The \diff{\lfloor2n'/3\rfloor}{2n'}-differential coloring of $G'$, in the case where $G$ is $3$-colorable.}
    \label{fig:reduction}
\end{figure}

Now, consider the case where $G'$ is $3$-colorable. Let $C_i
\subseteq V$ be the set of vertices of the input graph $G$  with
color $i$, $i=1,2,3$. Clearly, $C_1 \cup C_2 \cup C_3 = V$. We
compute a coloring $l$ of the vertices of graph $G'$ as follows (see
Fig.~\ref{fig:coloring}):
\begin{enumerate}[-]
  \item Vertices in $C_1$ are assigned colors from $1$ to $|C_1|$.
  \item Vertices in $C_2$ are assigned colors from $3n+|C_1|+1$ to $3n+|C_1|+|C_2|$.
  \item Vertices in $C_3$ are assigned colors from $5n+|C_1|+|C_2|+1$ to $5n+|C_1|+|C_2|+|C_3|$.
  \item For a vertex $v_1 \in V_1$ that is a neighbor of a vertex $v \in C_1$, $l(v_1) = l(v) + 2n$.
  \item For a vertex $v_1 \in V_1$ that is a neighbor of a vertex $v \in C_2$, $l(v_1) = l(v) - 2n$.
  \item For a vertex $v_1 \in V_1$ that is a neighbor of a vertex $v \in C_3$, $l(v_1) = l(v) - 4n$.
  \item For a vertex $v_2 \in V_2$ that is a neighbor of a vertexx $v_1 \in V_1$, $l(v_2) = l(v_1)+3n +|C_2|$.
\end{enumerate}
From the above, it follows that the color difference between
\begin{inparaenum}[(i)]
\item any two vertices in $G$,
\item a vertex $v_1 \in V_1$ and its neighbor $v \in V$, and
\item a vertex $v_1 \in V_1$ and its neighbor $v_2 \in V_2$,
\end{inparaenum}
is at least $2n = \lfloor\frac{2n'}{3}\rfloor$. Thus, $G'$ is
\diff{\lfloor 2n'/3\rfloor}{2n'}-differential colorable.
\end{proof}

%=================================================================
\section{An ILP for the Maximum k-Differential Coloring Problem}
\label{sec:ILP}
%=================================================================

In this section, we describe an integer linear program (ILP)
formulation for the maximum k-differential coloring problem. Recall
that an input graph $G$ to the maximum k-differential coloring
problem can be easily converted to an input to the maximum
1-differential coloring by creating a disconnected graph $G'$ that
contains all vertices and edges of $G$ plus $(k-1) \cdot n$ isolated
vertices. In order to formulate the maximum 1-differential coloring
problem as an integer linear program, we introduce for every vertex
$v_i \in V$ of the input graph $G$ a variable $x_i$, which
represents the color assigned to vertex $v_i$. The 1-differential
chromatic number of $G$ is represented by a variable $OPT$, which is
maximized in the objective function. The exact formulation is given
below. The first two constraints ensure that all vertices are
assigned colors from $1$ to $n$. The third constraint guarantees
that no two vertices are assigned the same color, and the forth
constraint maximizes the 1-differential chromatic number of the
graph. The first three constraints also guarantee that the variables
are assigned integer values.
$$
\begin{array}{llcllr}
\mbox{\textbf{maximize}}\:\: & OPT \\
\mbox{\textbf{subject to }} & x_i         &\le&  n  & \forall v_i\in V & \:\:\:\: \\
& x_i  &\ge&  1 \:\:& \forall v_i \in V
& \\
& |x_i - x_j|  &\ge&  1 \:\:& \forall (v_i,v_j) \in V^2
& \\
& |x_i - x_j|  &\ge&  OPT \:\:& \forall (v_i,v_j) \in E &
\end{array}
$$
Note that a constraint that uses the absolute value is of the form
$|X| \geq Z$ and therefore can be replaced by two new constraints:
\begin{inparaenum}[(i)]
\item $X + M \cdot b \ge Z$ and
\item $-X + M \cdot(1 - b) \ge Z$,
\end{inparaenum}
where $b$ is a binary variable and $M$ is the maximum value that can
be assigned to the sum of the variables, $Z + X$. That is, $M = 2n$. If $b$ is equal to
zero, then the two constraints are $X \ge Z$ and $-X + M \ge Z$,
with the second constraint always true. On the other hand, if
$b$ is equal to one, then the two constraints are $X + M \ge Z$ and
$-X \ge Z$, with the first constraint always true.

Next, we study two variants of the ILP formulation described above:
ILP-n and ILP-2n, which correspond to $k=1$ and $k=2$, and compare
them with GMap, which is a heuristic based on spectral methods
developed by Hu et al.~\cite{Gansner_Hu_Kobourov_2009_gmap}.

\begin{figure}[t!]
    \centering
    \begin{minipage}[b]{.32\textwidth}
        \centering
        \subfloat[\label{fig:figure11}{}]
        {\includegraphics[width=\textwidth]{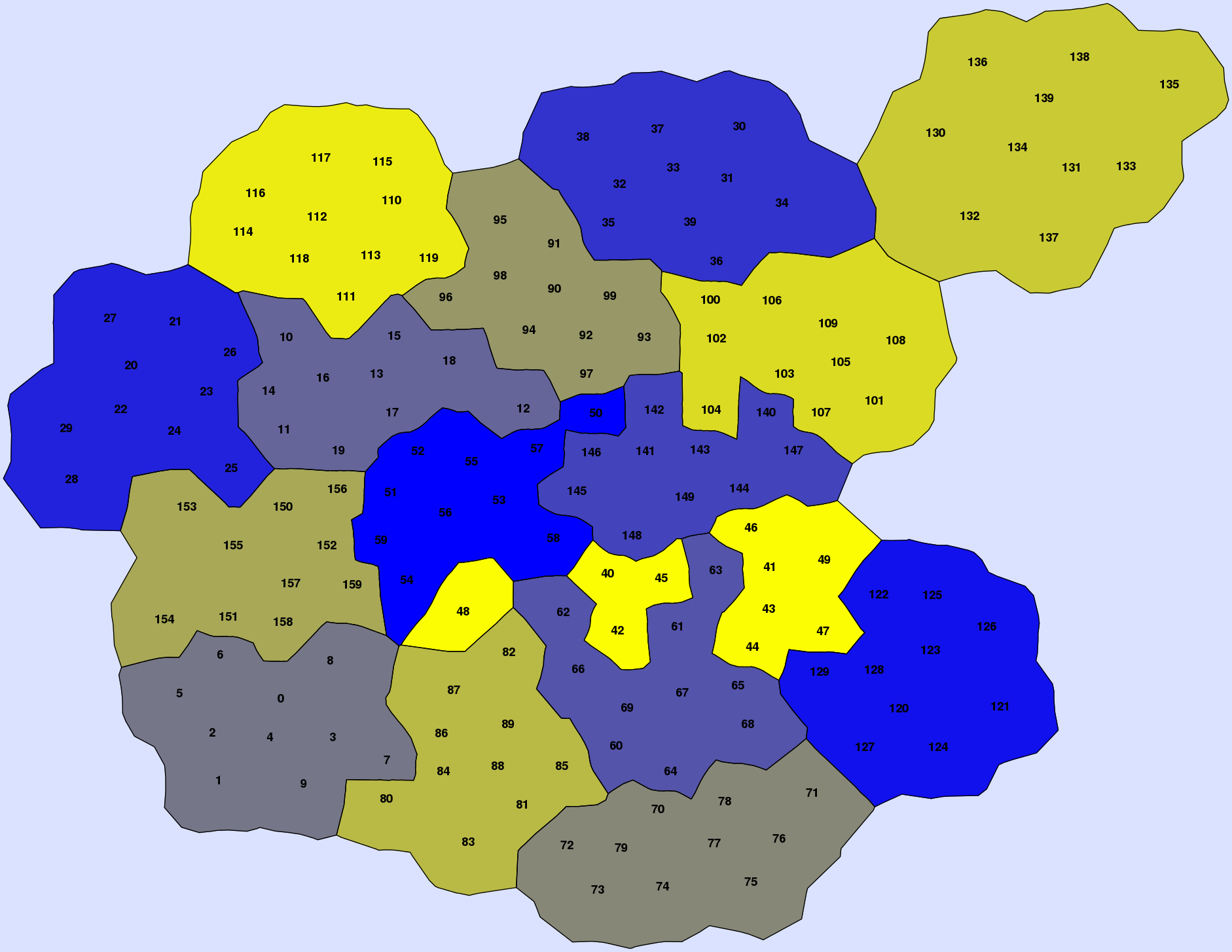}}
    \end{minipage}
    \begin{minipage}[b]{.32\textwidth}
        \centering
        \subfloat[\label{fig:figure12}{}]
        {\includegraphics[width=\textwidth]{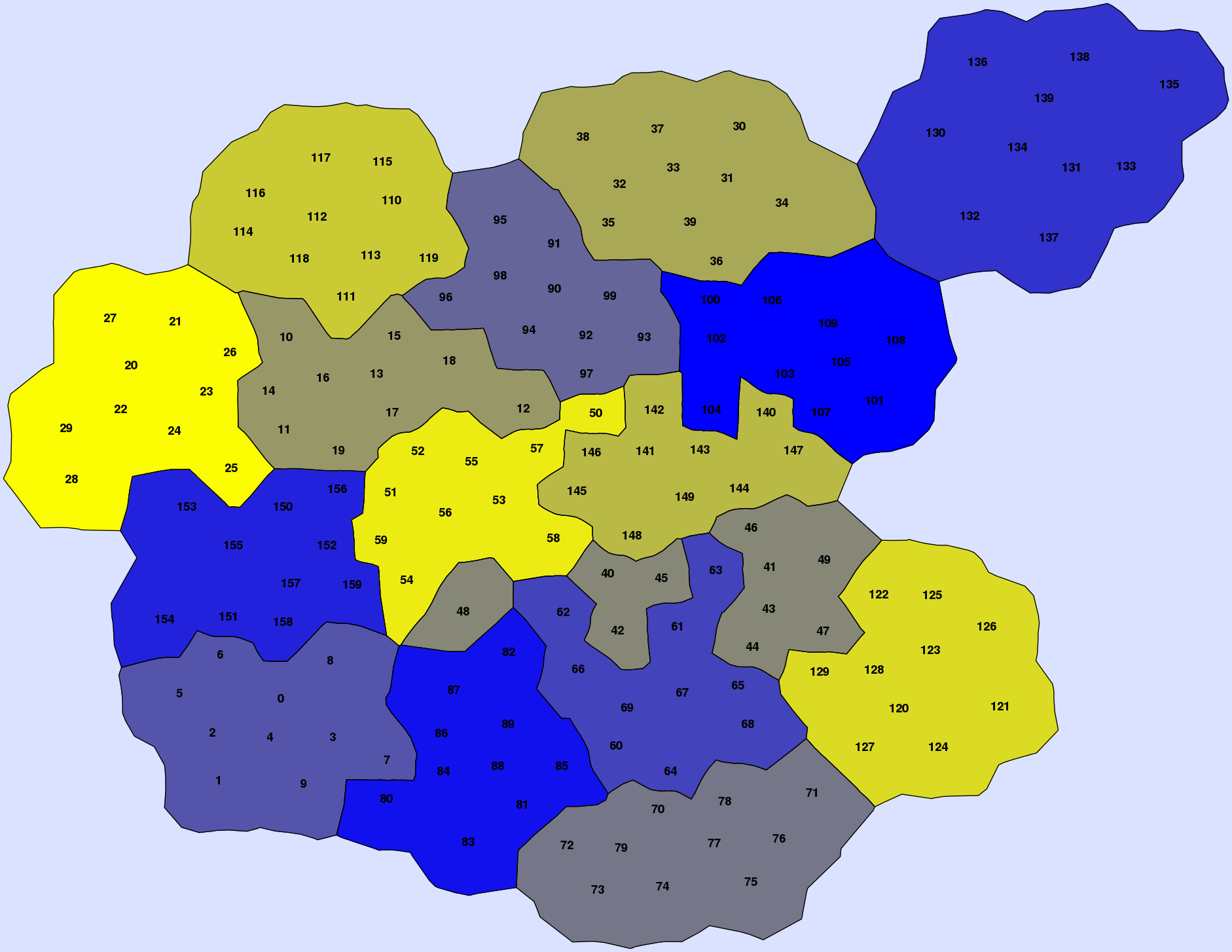}}
    \end{minipage}
    \begin{minipage}[b]{.32\textwidth}
        \centering
        \subfloat[\label{fig:figure14}{}]
        {\includegraphics[width=\textwidth]{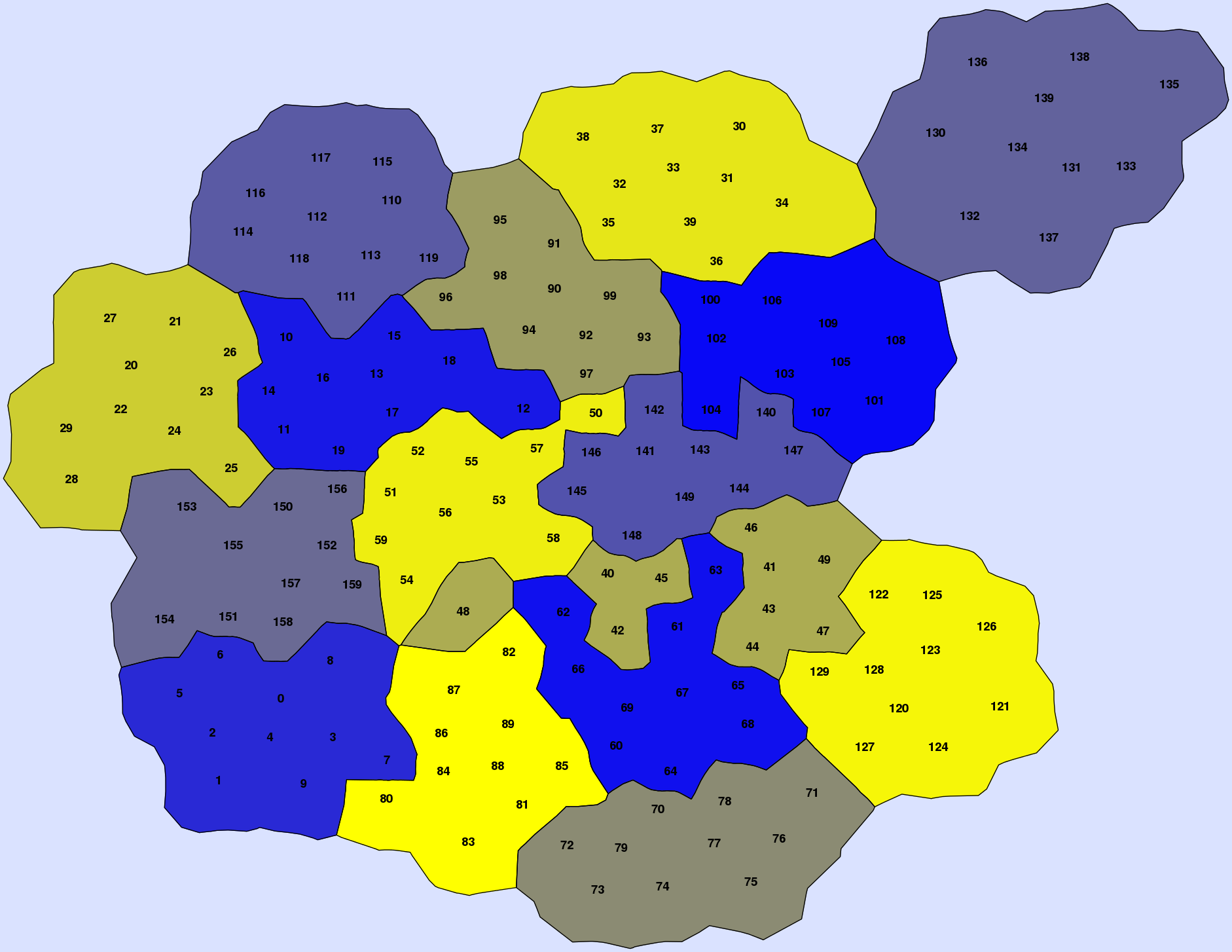}}
    \end{minipage}
    \caption{A map with 16 countries colored by: (a)~GMap~\cite{Gansner_Hu_Kobourov_2009_gmap}, (b)~ILP-n, (c)~ILP-2n.}
    \label{fig:samples}
\end{figure}

%=================================================================
\subsection{Experiment's Setup}
%=================================================================
We generate a collection of $1,200$ synthetic maps and analyze the
performance of ILP-n and ILP-2n,  on an Intel Core i5 1.7GHz
processor with 8GB RAM, using the CPLEX solver~\cite{cplex}. For
each map a country graph $G_c = (V_c,E_c)$ with $n$ countries is
generated using the following procedure.
\begin{enumerate}[(1)]
\item We generate $10n$ vertices and place an edge
between pairs of vertices (i,j) such that
$\lfloor\frac{i}{10}\rfloor = \lfloor\frac{j}{10}\rfloor$, with
probability 0.5, thus resulting in a graph $G$ with approximately $n$
clusters.
\item More edges are added between all pairs of vertices
with probability $p$, where $p$ takes the values $1/2, 1/4 \ldots
2^{-10}$.
\item Ten random graphs are generated for different values of $p$.
\item $G$ is used as an input to a map generating algorithm
(available as the \texttt{Graphviz}~\cite{graphviz01} function
{\texttt{gvmap}}), to obtain a map $M$ with country graph $G_c$.
\end{enumerate}

A sample map generated by the aforementioned procedure is shown in
Fig.~\ref{fig:samples}. Note that the value of $p$ determines the
``fragmentation'' of the map $M$, i.e., the number of regions in
each country, and hence, also affects the number of edges in the
country graph. When $p$ is equal to $1/2$, the country graph is a
nearly complete graph, whereas for $p$ equal to $2^{-10}$, the
country graph is nearly a tree.  To determine a suitable range for
the number of vertices in the country graph, we evaluated real world
datasets, such as those available at \url{gmap.cs.arizona.edu}. Even
for large graphs with over $1,000$ vertices, the country graphs tend
to be small, with less than $16$ countries.

%=================================================================
\subsection{Evaluation Results} \label{sec:app1}
%=================================================================
Fig.~\ref{fig:plots} summarizes the experimental results. Since $n$
is ranging from $5$ to $16$, the running times of both ILP-n and
ILP-2n are reasonable, although still much higher than GMap. The
color assignments produced by ILP-n and GMap are comparable, while
the color assignment of ILP-2n results in the best minimum color
distance. It is worth mentioning, though, that in the presence of
twice as many colors as the graph's vertices, it is easier to obtain
higher color difference between adjacent vertices. However, this
high difference comes at the cost of assigning pairs of colors that
are more similar to each other for non-adjacent vertices, as it is
also the case in our motivating example from the Introduction where
$G$ is a star.

\begin{figure}[t!]
    \centering
    \begin{minipage}[b]{\textwidth}
        \centering
        {\includegraphics[width=\textwidth]{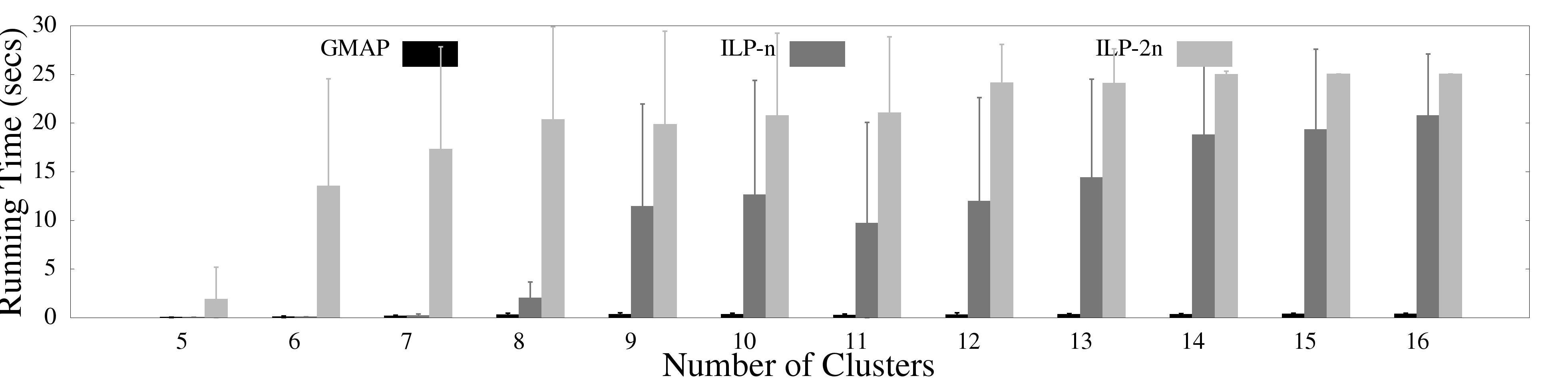}}
        {\includegraphics[width=\textwidth]{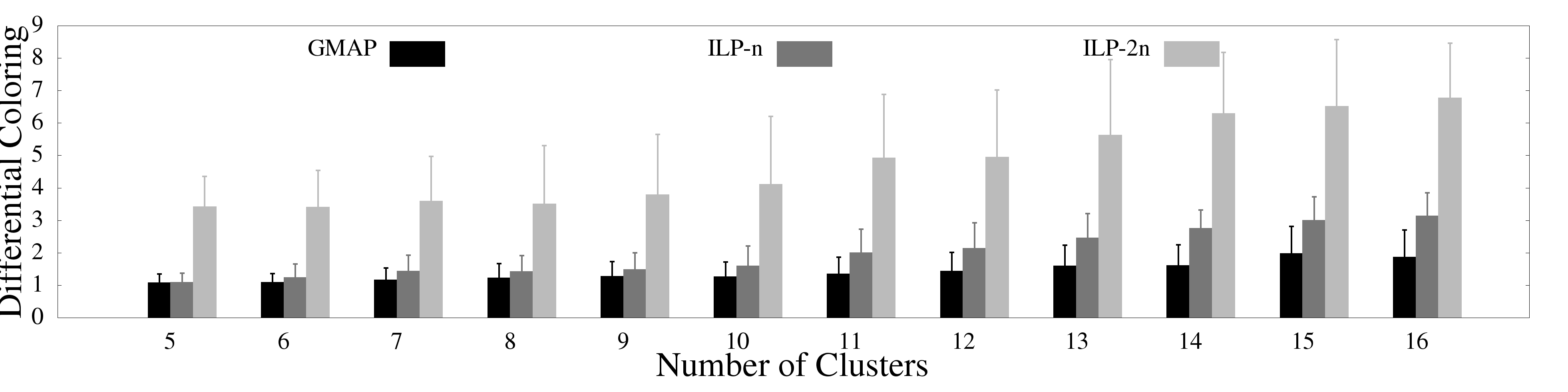}}
    \end{minipage}
    \caption{Illustration of:
    (a)~running time results for all algorithms of our experiment and
    (b)~differential coloring performance of algorithms GMap, ILP-n and ILP-2n.}
    \label{fig:plots}
\end{figure}

%=================================================================
\section{Conclusion and Future Work}
\label{sec:conclusion}
%=================================================================

Even though the \diff{2}{n}-differential coloring is NP-complete for
general graphs, in this paper, we gave a complete characterization
of bipartite, outerplanar and planar graphs that admit
\diff{2}{n}-differential colorings. Note that these
characterizations directly lead to polynomial-time recognition
algorithms. We also generalized the differential coloring problem
for more colors than the number of vertices in the graph and showed
that it is NP-complete to determine whether a general graph admits a
\diff{3}{2n}-differential coloring. Even for planar graphs, the
problem of determining whether a graph is
\diff{\lfloor2n/3\rfloor}{2n}-differential colorable remains
NP-hard. Several related problems are still open:

\begin{itemize}[-]
\item Is it possible to characterize which bipartite, outerplanar or
planar graphs are \diff{3}{n}-differential colorable?
\item Extend the characterizations for those planar graphs that
admit \diff{2}{n}-differential colorings to 1-planar graphs.
\item Extend the results above to \diff{d}{kn}-differential coloring
problems with larger $k > 2$.
\item As all planar graphs are \diff{\lfloor\frac{n}{3}\rfloor +
1}{2n}-differential colorable, is it possible to characterize which
planar graphs are \diff{\lfloor\frac{n}{3}\rfloor +
2}{2n}-differential colorable?
\item Since it is NP-complete to determine the $1$-differential
chromatic number of a planar graph~\cite{arxivplanar}, a natural
question to ask is whether it is possible to compute in polynomial
time the corresponding chromatic number of an outerplanar graph.
\end{itemize}

%=================================================================
\bibliographystyle{abbrv}
\bibliography{refs}
%=================================================================

\end{document}